\documentclass[10pt,conference]{IEEEtran}
\makeatother
\usepackage{graphicx}
\usepackage{tikz}
\usepackage{color}
\usepackage{amsmath, amsfonts, amssymb, mathrsfs}
\usepackage[amsmath,thmmarks]{ntheorem}
\usepackage{subfigure}
\usepackage{epstopdf}

\usepackage{pgfplots}
\pgfplotsset{compat=newest}
\usetikzlibrary{plotmarks}
\usetikzlibrary{arrows.meta}
\usepgfplotslibrary{patchplots}
\usepackage{grffile}

\usepackage{arydshln}
\usepackage{mathabx}
\usepackage{cite}
\usepackage[super]{nth}
\usepackage{cite}

\usepackage{algorithm}
\usepackage[noend]{algpseudocode}
\usepackage{bbm}
\usepackage{bm}
\usetikzlibrary{calc}
\usetikzlibrary{arrows}
\usetikzlibrary{patterns}
\usetikzlibrary{shapes}


\theoremstyle{plain}
\theoremheaderfont{\itshape\bfseries}
\theorembodyfont{\normalfont\itshape}
\theoremseparator{:}
\theoremsymbol{}
\newtheorem{theorem}{Theorem}
\newtheorem{lemma}[theorem]{Lemma}

\newtheorem{proposition}[theorem]{Proposition}

\newtheorem{definition}{Definition}

\theoremstyle{nonumberplain}

\theoremstyle{plain}
\theoremheaderfont{\itshape\bfseries}
\theorembodyfont{\normalfont}
\theoremseparator{:}
\theoremsymbol{}

\theoremstyle{plain}
\theoremheaderfont{\itshape\bfseries}
\theorembodyfont{\normalfont}
\theoremseparator{:}

\newtheorem{assum}{Assumption}

\theoremstyle{nonumberplain}
\theoremsymbol{\rule{1ex}{1ex}}

\newtheorem{proof}{Proof}
\newlength\fheight
\newlength\fwidth
\newcommand{\norm}[1]{\| #1 \|}
\newcommand{\inn}[2]{\langle #1,#2 \rangle}
\newcommand{\lrbrace}[1]{\left\{ #1 \right\}}

\newcommand{\normi}[1]{{\left\vert\kern-0.25ex\left\vert\kern-0.25ex\left\vert #1 
		\right\vert\kern-0.25ex\right\vert\kern-0.25ex\right\vert}}
\newcommand{\inni}[2]{{\langle\kern-0.25ex\langle #1,#2
		\rangle\kern-0.25ex\rangle}}

\def\boldnormtheta{\boldsymbol{\theta}}

\def\X{\mathcal{X}}

\def\rmf{ \mathrm{f} }

\def\nat{ \mathbb{N} }								

\def\real{ \mathbb{R} }								

\def\A{\mathcal{A}}

\def\P{\mathcal{P}}

\def\rmq{\mathrm{q}}

\def\bs{ \mathbf{s} }									

\def\bC{ \mathbf{C} } 								
 
\def\bU{ \mathbf{U} }
\def\bq{ \mathbf{q} }  

\def\bss{ \boldsymbol{s} }						
\def\bsw{ \boldsymbol{w} }						
\def\bsx{ \boldsymbol{x} }						
\def\bsp{ \boldsymbol{p} }
\def\bsu{ \boldsymbol{u} }
\def\bss{ \boldsymbol{s} }
\def\bsg{ \boldsymbol{g} }
\def\bsM{ \boldsymbol{M} }
\def\bsy{ \boldsymbol{y} }						

\def\U{ \mathcal{U} }											

\DeclareMathOperator{\dom}{\text{dom}}
\DeclareMathOperator{\relint}{\text{relint}}

\DeclareMathOperator{\SINR}{\text{SINR}}
\newcommand\restr[2]{{
  \left.\kern-\nulldelimiterspace 
  #1 
  \right|_{#2} 
  }}

\DeclareMathOperator*{\argmax}{arg\,max}

\title{Welfare Measure for Resource Allocation with Algorithmic Implementation: Beyond Average and Max-Min}
\author{\IEEEauthorblockN{Ezra Tampubolon and Holger Boche}
	
\IEEEauthorblockA{
  Lehrstuhl f{\"u}r Theoretische Informationstechnik\\
  Technische Universit{\"a}t M{\"u}nchen, 80290 M{\"u}nchen, Germany\\
  \{ezra.tampubolon,boche\}@tum.de}

 }

\begin{document}
%
\maketitle
%
\begin{abstract}
In this work, we propose an axiomatic approach for measuring the performance/welfare of a system consisting of concurrent agents in a resource-driven system. Our approach provides a unifying view on popular system optimality principles, such as the maximal average/total utilities and the max-min fairness. Moreover, it gives rise to other system optimality notions that have not been fully exploited yet, such as the maximal lowest total subgroup utilities. For the axiomatically defined welfare measures, we provide a generic gradient-based method to find an optimal resource allocation and present a theoretical guarantee for its success. Lastly, we demonstrate the power of our approach through the power control application in wireless networks.    
\end{abstract}
\section{Introduction}
The field of resource allocation (RA) concerns with the assignment of available scarce resources to various agents in a system characterized by competitive environment. The objective of the system manager in this regard is to lead the population toward an optimal state. RA has been an inherent part in widespread applications in economics, operation research, and engineering.

For instance, RA is a indispensable part of wireless systems, as they require a fundamental and crisp understanding of design principles and control mechanisms to manage network resources efficiently. RA policies in those systems aim to maximize the Quality of Service (QoS) at the user level, and to ensure efficient and optimized operation at the network level by maximizing the operator's revenue. RA management in wireless communications may include a broad spectrum of network functionalities, such as scheduling, transmission rate control, power control, bandwidth reservation, call admission control, transmitter assignment, and handover \cite{StanczakBoche2009,LeeChuah2014,Ahmed2005}.

One popular principle of system optimality/welfare in wireless systems is the so-called max-average utility principle. Following this principle, the RA strategy consists of solving an optimization problem having the average of the users'/links' utilities as the objective (see e.g. \cite{StanczakBoche2009,Mach2015,Kluegel2018,Matthiesen2020}). 
Another popular principle of system optimality/welfare in wireless systems is the so-called max-min fairness \cite{Bertsekas1992}. This principle strives for welfare on the basis of the worst-off links/users (see e.g., \cite{Zheng2016,Sadeghi2018}). RA policies following this principle usually lead to the equal allocation so that it is not possible to increase any utilities without deterioting others that are smaller or equal. Both, max-average utility and max-min fairness principle is also an object of study in other fields of sciences such as in economics, where they are known as the utilitarian and the Rawlsian welfare principle \cite{Rawls1971}.

Those optimality principles have advantages and disadvantages. The max-average utility principle ensures the achievement of the optimal (total) system performance of cost of reduction of some of the agents \cite{Massoulie2002}. On the contrary, the max-min fairness endeavors to optimize the "weakest" agents at the expense of a considerable drop in system efficiency expressed in terms of total service. A common approach is to balance between both optimality principles by replacing the agents' utilities in the formulation of the max-average optimization problem \cite{Kelly1998,Massoulie2002,Mo2000}, or by optimizing the average utility given a target fairness \cite{Zabini2017}.

  This paper's contribution is the introduction of a general concept for measuring a systems' performance/welfare, including the usually used measures such as the average agents' utilities and the minimum of agents' utilities. Our approach based on the duality theory allows one to relate a specific welfare measure to the average welfare via the averaging weights. This relation gives rise to the utilization (super)gradient of agents' utilities to find the optimal allocation. We also present the guarantee of the success, and illustrate the supergradient algorithm's applicability in achieving the corresponding welfarism in a cellular network. Thereby, we focus on the class of welfare measure called the low $K$-average welfare, where $K$ is a number smaller than the total number of the agents, measuring the smallest average of $K$ agents. Besides the optimality on its term, i.e., it maximizes the total utilities of $K$ worst-off agents, we show that the optimization respective to the low $K$-average welfare offers a natural tradeoff between average optimality and max-min fairness. Due to the space limitations, we skip the proofs of the theoretical results in the main part of this paper and provide them in the appendix.  
  
\paragraph*{Basic Notations and Notions}
Given a vector $\bsx\in\real^{N}$. Unless otherwise stated, $\bsx^{(i)}$ denotes the $i$'th element of $\bsx$. The notion of concavity and properness of functions used in this work coincides with the notion given in the literature (see e.g., \cite{Rockafellar1970}). Let be $\bsg\in\real^{D}$ and $x\in\real^{D}$. We say $\bsg$ is a supergradient of $\rmf$ at $\bsx\in\real^{D}$ if $\rmf(\bsy)\leq \rmf(\bsx)+\inn{\bsg}{\bsy-\bsx}$, $\forall \bsy\in\real^{N}$.
We denote the set of the supergradient of $\rmf$ at $\bsx$ by $\partial \rmf(\bsx)$. Sometimes, we write the later as $\partial_{\bsx} \rmf(\bsx)$ to emphasize the variable at which we consider the supergradient. For a given function $\rmf$, $\rmf_{e}$ denotes the function $\rmf(e^{(\cdot)})$. $\Delta([D])$ denotes the simplex on $\real^{D}$.

\section{System Model and Problem Formulation}
\label{Sec:ajjaasggsgsgsss}

\paragraph*{Allocation Vector and Agent's Utility}We consider the problem of allocating $D$ resources in a system consisting of $N$ agents. We denote the amount of allocation of resources by a vector $\boldnormtheta\in\real^{D}$, where $\boldnormtheta^{(k)}$ stands for the amount of allocation of the resource $k\in [N]$. The benefit of agent $i\in [N]$ from the allocation $\boldnormtheta\in\real^{D}$ is measured by a function $\bU^{(i)}:\real^{D}\rightarrow\overline{\real}$, called the \textit{utility function}, and is specifically given by $\bU^{(i)}(\boldnormtheta)$. Throughout this work, we usually assume $\bU^{(i)}$ to be concave and proper. Usually, the choice of resources allocation is restricted by some practical considerations, such as budget restrictions. This occurence can be modeled by a subset $\Theta\subset\real^{N}$ called the \textit{feasible set}, from which an allocation $\boldnormtheta$ can be chosen.

\paragraph*{Wireless Network Application}
For the purpose of illustrations, we consider in this paper the specific application of wireless network. We assume that the network contains $N$ users/agents 
 transmiting their independent data concurrently (due to transmission interferences). In this setting, the resource to be allocated is the (log of the) transmit power of each of the users. A resource allocation policy in this context is a vector $\bss\in\real^{N}$ specifying the \textit{log transmit power} of the users. To be specific, for any $k\in[N]$, $\bss^{(k)}\in\real$ specifies the log transmit power of agent $k$, i.e., $e^{\bss^{(k)}}$ is the transmit power of user $k$. Considering log transmission power is a usual practice in power control as it reveals the hidden convexity in the corresponding optimization problem \cite{TanPalomar2007}.Now, due to power constraints, we require that $e^{\bss}\in\P$, where:
\begin{equation}
\label{Eq:jjaaggsffsfsfsgsgsgsss}
\P:=\lrbrace{\bsp\in\real^{N}_{\geq 0}:~\mathbf{C}\bsp\leq \hat{\bsp}},
\end{equation}
for some a given $\bC\in\real_{\geq 0}^{M\times N}$ and $\hat{\bsp}\in\real^{M}$, with $M\in\nat$. Consequently, the set $\log\P$, where $\log$ is understood elementwise, corresponds to the feasible set $\Theta$ in the general resource allocation setting.

 For any $k\in [N]$, one can measure the utility of agent $k$ by the so called \textit{Quality-of-Service (QoS)} value. This value is specified by the so-called \textit{signal-to-interference-noise ratio (SINR)} of agent $k$ given by (see e.g., Chapter 4 in \cite{StanczakBoche2009}):
\begin{equation}
\label{Eq:jaajjssgsgsgsgsfsfssss}
\begin{split}
\SINR^{(k)}_{e}(\bss):=\SINR^{(k)}(e^{\bss})&:=\tfrac{V_{k,k}e^{\bss^{(k)}}}{\sum_{\substack{l\in [N]\\l\neq k}}V_{k,l}e^{\bss^{(l)}}+\sigma_{k}^{2}}
\end{split}
\end{equation} 
The constant $V_{k,k}>0$ in above definition represents the user $k$'s communication gain by power utilization. For any $k,l\in [N]$ with $k\neq l$, $V_{k,l}\geq 0$ specifies $k$'s performance reduction caused by link $l$ communication activity in form of interferences. The constant $\sigma_{k}^{2}$ in above definition denotes the power of the noise in agent $k$'s receiver. Finally by the definition of the SINR provided before, we can specify the QoS of agent $k$ by:
\begin{equation}
\label{Eq:aajajssgsgsgsfsfsgssffsfsss}
\bq^{(k)}_{e}(\bss):=\bq^{(k)}(e^{\bss})=\psi(\SINR^{(k)}(\bss)),
\end{equation}
where $\psi:\real\rightarrow\overline{\real}$ is desired to satisfy the following:
\begin{assum}
	\label{Ass:Log}
	$\psi:\real\rightarrow\overline{\real}$	is a function with $\dom(\psi)\subseteq \real_{>0}$ such that $\psi_{e}:=\psi(e^{(\cdot)})$ is concave.
\end{assum}
The condition above is of technical nature and allows one to utilize convex optimization method for solving QoS optimization as the QoS is concave in the logarithmic of power, i.e. $\bq^{(k)}_{e}$ is concave (see e.g. Theorem 7 in \cite{Naik2011}). An example of $\psi$ satisfying Assumption \ref{Ass:Log} is $\psi(x)=\log(x)$. For this kind of $\psi$ the QoS corresponds to the Shannon's theoretical data rate in the high SINR regime with normalized bandwidth $B=1$ as $\psi(x)$ approximates $\log(1+x)$ for large $x$. However, notice that $\psi(x)=\log(1+x)$ does not satisfy Assumption \ref{Ass:Log}. Another possible choice for $\psi$ satisfying Assumption \ref{Ass:Log} is $\psi(x)=-1/x^{\alpha}$, where $\alpha\geq 1$. This choice yields the QoS interpretable as the negative of bit error approximation for diversity order $\alpha$.

\paragraph*{Welfare Maximization} Let us now go back to the general setting of RA. The practice of RA is to find an allocation vector which yields the maximal possible benefit for the system. 
It is usual practice to find a compromise solution, i.e., an allocation such that the increase the utility of an agent leads to a simultaneous decrease of at least one other's utility.  
Such a compromise solution is specified by the whole system's performance as a function of the utility perceived by each partaker, according to the purpose of the system. One popular way to measure the utility of a system is by taking the average of the utilities of the individuals in the system. This definition of system utility corresponds to the so called average optimality principle or utilitarian welfare.
Respective to this principle, the problem to solve is:
\begin{equation}
\label{Eq:ajajjsgsgsfsffsfggsssssss}
\min_{\boldnormtheta\in\Theta} \frac{1}{N}\sum_{i=1}^{N}\bsw^{(i)}\bU^{(i)}(\boldnormtheta),
\end{equation}
where $\bsw^{(i)}$, $i\in [N]$, is a sequence of non-negative scalars summing up to one. Another popular way to measure the utility of a system is by taking the minimum of the utilities of the individuals in the system. This definition of system utility is also known as the max-min fairness or Rawlsian welfare. The corresponding problem to solve is:
\begin{equation}
\label{Eq:jaajjagsgsgsfsfssgsfssfsss}
\max_{\boldnormtheta\in\Theta} \min_{i\in [N]}\bU^{(i)}(\boldnormtheta)
\end{equation}
As discussed in the introduction, both the above-presented optimality principles have advantages and disadvantages. In the literature  \cite{Kelly1998,Massoulie2002,Mo2000}, it is usual practice to balance between those optimality principles. In contrast to the prior work, we aim in this work to achieve this balance by finding an approach unifying \eqref{Eq:ajajjsgsgsfsffsfggsssssss} and \eqref{Eq:jaajjagsgsgsfsfssgsfssfsss}, since the corresponding abstract concept might gives rise to another alternative optimality principle inbetween the aforementioned principles. Furthermore, our requirement for the desired approach is that it should allow one to use a generic method such as the first-order method to achieve the corresponding optimal allocation. 

\section{Welfare Measure: Axiomatic Approach, Robust Representation, and Supergradient}
\label{Sec:ajajgsgsgsffssgsfsgsfgssss}
As discussed in the previous section, optimal resource allocation strategy requires a measure for the system-wide performance. Usually used measure is the so-called \textit{average utility}, which takes the average of the individual welfares/utility:
\begin{equation*}
\overline{\Phi}^{\bsw}:\real^{N}\rightarrow\real,\quad \bsu\mapsto\sum_{i=1}^{N}\bsw^{(i)}\bsu_{(i)},
\end{equation*}
where $\bsw\in\Delta([N])$ is a given weight. Usually, one chooses equal weights. However, it is convenient to choose other weights in order to involve several technical aspects, such as the priority of the users. Average utility gives rise to the utilitarian welfare principle \eqref{Eq:ajajjsgsgsfsffsfggsssssss}. Another popular welfare measure is the so-called \textit{minimum utility}:
\begin{equation*}
\underline{\Phi}:\real^{N}\rightarrow\real,\quad \bsu\mapsto \min_{i\in [N]}\bsu^{(i)}.
\end{equation*}
In contrast to the average utility, this functional measures the system's welfare by considering the minimum individual welfare/utility. The minimum welfare gives rise to the Rawlsian welfare principle given in \eqref{Eq:jaajjagsgsgsfsfssgsfssfsss}.


We provide in the following the general concept of the welfare measure:
\begin{definition}[Welfare Measure (WM)]
Let be $\Phi:\real^{N}\rightarrow\real$. We say $\Phi$ is a welfare measure (WM) if $\Phi$ satisfies the following:
\begin{itemize}
	\item{(A1)} $\Phi$ is monotonic, i.e., $\Phi(\bsu)\leq\Phi(\bsu')$ if $\bsu\leq \bsu^{'}$.  
	\item{(A2)} $\Phi$ is concave.
	\item{(A3)} $\Phi$ is positively homogeneous, i.e., $\Phi(\lambda \bsu)=\lambda \Phi(\bsu)$, $\forall \lambda\geq 0,~\bsu\in\real^{N} $.  
	\item{(A4)} $\Phi$ is upper semi-continuous, i.e., for any $\bsu_{0}\in\real^{N}$,$\limsup_{\bsu\rightarrow \bsu_{0}}\Phi(\bsu)\leq \Phi(\bsu)$.
	\item{(A5)} $\Phi(1)=-\Phi(-1)$  
\end{itemize}
\end{definition}

The concavity condition (A3) and the upper semi-continuity condition (A4) allow us to use the concept of conjugate function in convex analysis \cite{Rockafellar1970,BoydBook2004} for analyzing welfare measures. The central result regarding this concept of our benefit is the so-called Fenchel-Moreau Theorem. Fenchel-Moreau Theorem allows one to write a fairly general concave function $\rmf$ as the maximum of a penalized linear function over an uncertainty set. The corresponding penalty function is given by the concave conjugate $\rmf^{*}(y):=\min_{x\in\real^{D}}\lrbrace{\inn{x}{y}-\rmf(x)}$. Furthermore, the homogeneity condition (A3) allows us to neglect the penalty function, and finally the monotonicity condition (A1) and the condition (A5) help us specify the corresponding uncertainty set. Our result is given specifically in the following theorem:
\begin{theorem}
	\label{Thm:ajajssgsgsffsfsssssss}
Let $\Phi:\real^{N}\rightarrow\real$ be a function. Then:
\begin{enumerate}
\item  $\Phi$ is a welfare measure if and only if it can be represented by:
\begin{equation}
\label{Eq:ajajajshsshsggsgsgsgsssss}
\Phi(\bsu)=\min_{\bsw\in\mathcal{U}}\overline{\Phi}^{\bsw}(\bsu),\quad\forall \bsu\in\real^{N}
\end{equation} 
where $\U$ is a non-empty closed and convex subset of the simplex $\Delta[N]$. 
\item $\U$ in \eqref{Eq:ajajajshsshsggsgsgsgsssss} is uniquely given by $\partial\Phi(0)$.  
\end{enumerate}   
\end{theorem}

Theorem \ref{Thm:ajajssgsgsffsfsssssss} relates a general welfare measure to the usually used average utility, as it asserts that any welfare measure can be written as the maximum of the average welfare respective to the weights/priorities. 
Moreover, Theorem \ref{Thm:ajajssgsgsffsfsssssss} specifies the corresponding set of weights in the optimization program as the set of the welfare measure's supergradient at point $0$.

To use the supergradient method for welfare maximization, we need to compute a supergradient of the welfare measure. The following lemma based on the representation \eqref{Eq:ajajajshsshsggsgsgsgsssss} provides the corresponding tool:
\begin{lemma}
	\label{Lem:aajssgsgsgsfsfsss}
	Let be $\bsu\in\real^{D}$ and $\Phi$ be a welfare measure. $\bsw\in\real^{D}$ is a supergradient of $\Phi$ at $\bsu$ if and only if $\bsw\in\partial\Phi(0)$ and $\inn{\bsw}{\bsu}=\Phi(\bsu)$  	
\end{lemma}
Above lemma specifies the problem of finding a supergradient of a welfare measure to the problem of finding the solution $\bsw\in\partial\Phi(0)$ of the equation $\Phi(\bsu)=\inn{\bsw}{\bsu}$ for a fixed $\bsu$. For later purpose, we denote the set of such solutions by $\mathcal{W}(\Phi,\bsu)$, i.e.:
\begin{equation}
\label{Eq:ajajajsgsgsgsgssfsfsggsfsfsss}
\mathcal{W}(\Phi,\bsu):=\lrbrace{\partial\Phi(0):~\inn{\bsw}{\bsu}=\Phi(\bsu)}.
\end{equation}
\section{Supergradient Method for Welfare Optimization}
\label{Sec:SupGrad}
In this section, we aim to solve the optimization problem:
\begin{equation}
\label{Eq:ajajashssgggshsgsgssss}
\max_{\boldnormtheta\in\Theta}\Phi(\bU(\boldnormtheta)),
\end{equation}
where $\Phi$ is a welfare measure, $\bU$ is a vector-valued function specifying agents' utilities, and $\Theta$ is a problem-specific constraint set. One canonical way to solve the problem having the form \eqref{Eq:ajajashssgggshsgsgssss} is the so-called projected supergradient method whose iterate is given by:
\begin{equation}
\label{Eq:jajshgsgsggsfsfsffsgss}
\boldnormtheta_{t+1}=\Pi_{\Theta}(\boldnormtheta_{t}+\gamma_{t}g_{t}),
\end{equation} 
where $\Pi_{\Theta}$ denotes the usual Euclidean projection, $\gamma_{t}>0$ is a given step-size, and $g_{t}$ is a supergradient of $\Phi(\bU)$ at the resource allocation $\boldnormtheta_{t}$. 

To implement the supergradient method \eqref{Eq:jajshgsgsggsfsfsffsgss}, we need to ensure that $\Phi(\bU)$ is convex and to compute at each step $t$ a supergradient of $\Phi(\bU)$ at the iteration point $\boldnormtheta_{t}$. For this purposes, we can utilize the following consequence of Lemma
 \ref{Lem:aajssgsgsgsfsfsss}:
\begin{theorem}
	\label{Thm:auauasgsgsfsfsssss}
	Let $\Phi$ be a welfare measure, and for all $i\in [N]$, $\bU^{(i)}:\real^{D}\rightarrow\overline{\real}$ be a proper concave function. Suppose that $\boldnormtheta\in\bigcap_{i=1}^{N}\relint(\dom(U_{i}))\neq \emptyset$.  Then $\Phi(\bU)$ is a proper concave function. Furthermore, let be $\boldnormtheta\in\bigcap_{i=1}^{N}\relint(\dom(U_{i}))$, $\tilde{\bsg}^{(i)}\in\partial \bU^{(i)}(\boldnormtheta)$, $i\in [N]$, and $w\in\partial\Phi(0)$ satisfying:
	\begin{equation}
	\label{Eq:aahahsgsgsfsfssgfsssgsgs}
	\Phi(\bU(\boldnormtheta))=\inn{\bsw}{\bU(\boldnormtheta)}.
	\end{equation}
	Then:
	\begin{equation}
	\label{Eq:ajajajssggssgffsfssss}
	\sum_{i=1}^{N}\bsw^{(i)}\tilde{\bsg}^{(i)}\in \partial_{\boldnormtheta}\Phi(\bU(\boldnormtheta)).
	\end{equation}
\end{theorem}
 To compute a supergradient of $\Phi(\bU)$ at a point $\boldnormtheta$, we first query the (super)gradients of the utilities at the resource allocation of our interest. Finally, we obtain the supergradient of $\Phi(\bU)$ at point $\boldnormtheta$ by averaging the latter objects with weights contained in the set  $\mathcal{W}(\Phi,\bU(\boldnormtheta))$ defined in \eqref{Eq:ajajajsgsgsgsgssfsfsggsfsfsss}. Technically, one needs, in order to choose a weight in $\mathcal{W}(\Phi,\bU(\boldnormtheta))$, to know $\partial\Phi(0)$ and solve a corresponding linear equation. However, this is an easy task at least for the specific class of welfare measure discussed later in this paper (Section \ref{Sec:Case}). 
Finally, we provide the specific supergradient algorithm for solving the optimization problem \eqref{Eq:ajajashssgggshsgsgssss} in Algorithm \ref{Alg:aoaishhjddhhddddeee2}. We refer the corresponding algorithm throughout this work as pupergradient method for welfare maximization (SMWM). 

%

%



The first step to guarantee the success of SMWM to ensure that the sequence of the supergradients produced by SMWM. This is necessary in order to eliminate the possibility that the corresponding dupergradient method alternates around the solution of the corresponding optimization problem. Provided that the utilities of the agents have uniformly bounded supergradients, this crucial condition is fulfilled: 
  
\begin{lemma}[Boundedness of Supergradients for SMWM]
	\label{Lem:bounded}
Suppose that for any $i\in [N]$, the superdifferential set of $\bU^{(i)}$ is uniformly bounded, i.e.:
\begin{equation}
\label{Eq:ajajajssgsgsfsfsffssssss}
\bsM^{(i)}:=\sup_{\substack{\boldnormtheta\in\Theta\\\bsg\in\partial \bU^{(i)}(\boldnormtheta)}}\norm{\bsg}<\infty.
\end{equation}	
Then for $\bsg_{t}$, $t\in [T-1]$ given in SMWM \eqref{Eq:Supergrad}, it holds:
\begin{equation*}
\sup_{t\in [T-1]_{0}}\norm{\bsg_{t}}_{2}\leq \max_{\bsw\in\partial\Phi(0)}\inn{\bsw}{\bsM}<\infty
\end{equation*}
\end{lemma}
	\begin{algorithm}[htbp]
	\caption{Supergradient Method for Welfare Maximization (SMWM)}
	\label{Alg:aoaishhjddhhddddeee2}
	\begin{algorithmic}[1]
		\Require 
		Initial iterate $\boldnormtheta_{0}\in\real^{D}$, time horizon $T\in\nat$, step-size sequence $(\gamma_{t})_{t\in [T-1]_{0}}$.
		\For{$t=0,\ldots,T-1$}
		\State Observe $\Phi(\bU(\boldnormtheta_{t}))$
		\State Choose a weight vector $\bsw_{t}\in\partial\Phi(0)$ satisfying:
		\begin{equation*}
		\Phi\left( \bU^{(1)}(\boldnormtheta_{t}),\ldots,\bU^{(N)}(\boldnormtheta_{t})\right) =\sum_{i=1}^{N}\bsw^{(i)}_{t}\bU^{(i)}(\boldnormtheta)
		\end{equation*}
		\For{all agents $i\in [N]$}
		\State Query $\tilde{\bsg}_{t}^{(i)}\in\partial \bU^{(i)}(\boldnormtheta_{t})$ from agent $i$
		\EndFor
		
		\State Accumulate agents' supergradients:
		\begin{equation}
		\label{Eq:Supergrad}
		\bsg_{t}=\sum_{i=1}^{N}\bsw^{(i)}_{t}\tilde{\bsg}^{(i)}_{t}
		\end{equation}
		\State Updates the resource allocation:
		\begin{equation}
		\label{Eq:jajjghssggsgsgsfsfsgsgsfsfss}
		\boldnormtheta_{t+1}\gets\Pi_{\Theta}\left( \boldnormtheta_{t}+\gamma_{t}\bsg_{t}\right) 
		\end{equation}
		\EndFor
		\State Take the ergodic average of $(\boldnormtheta_{t})_{t\in [T-1]_{0}}$ or the best iterate:
		\begin{equation}
		\label{Eq:ajajajssshsgsgshsgsgshgs}
		\overline{\boldnormtheta}^{\gamma}_{T}=\tfrac{\sum_{t=0}^{T-1}\gamma_{t}\boldnormtheta_{t}}{\sum_{t=0}^{T-1}\gamma_{t}}\quad\text{or}\quad\boldnormtheta_{\max,T}\in\argmax_{t\in [T-1]_{0}}\Phi(\bU(\boldnormtheta_{t}))
		\end{equation}
	\end{algorithmic}
\end{algorithm}
\setlength{\textfloatsep}{0pt}
Now, we can provide a guarantee for the success of SMWM in finding the solution of \eqref{Eq:ajajashssgggshsgsgssss}:
\begin{theorem}[Convergence of SMWM]
	\label{Thm:akakasshshssgsggs}
	
Suppose that \eqref{Eq:ajajajssgsgsfsfsffssssss} holds,
and that $D_{\Theta}^{2}:=\max_{\boldnormtheta,\tilde{\boldnormtheta}\in\Theta}\tfrac{\norm{\boldnormtheta-\tilde{\boldnormtheta}}_{2}^{2}}{2}<\infty$.
Then for the output $\tilde{\boldnormtheta}$ (see \eqref{Eq:ajajajssshsgsgshsgsgshgs}) and $\boldnormtheta_{*}$ a solution of \eqref{Eq:ajajashssgggshsgsgssss}, it holds:  
\begin{equation*}
\Phi(\bU(\tilde{\boldnormtheta})\geq \Phi(U(\boldnormtheta_{*}))- \tfrac{D_{\X}^{2}}{\sum_{t=0}^{T-1}\gamma_{t}}- \tilde{M}^{2}\tfrac{\sum_{t=0}^{T-1}\gamma_{t}^{2}}{\sum_{t=0}^{T-1}\gamma_{t}},
\end{equation*}
where $\tilde{M}$ is given by $\tilde{M}:=\max_{\bsw\in\partial\Phi(0)}\inn{\bsw}{\bsM}$
\end{theorem}
The proof of above result follows from Lemma \ref{Lem:bounded} and standard proof of the convergence of supergradient method (see e.g. Theorem 3.1 in \cite{lan2020first}).

Theorem \ref{Thm:akakasshshssgsggs} gives a guideline for an appropriate choice of step size sequences. For instance, setting the step size $\gamma_{t}=\sqrt{\tfrac{2D_{\X}^{2}}{TM^{2}}}$, $t\in [T-1]_{0}$, it follows that for small $\epsilon>0$, we need $T\geq\sqrt{MD_{\X}/(\sqrt{2}\epsilon)}$ steps in order that the output $\tilde{\boldnormtheta}_{T}$ of SWMW satisfies $\Phi(\bU(\tilde{\boldnormtheta}_{T}))\geq\max_{\boldnormtheta\in\Theta}\Phi(\bU(\boldnormtheta))-\epsilon$.
We may alternatively choose the variable step-size such as $\gamma_{t}=\sqrt{\tfrac{2D_{\X}^{2}}{tM^{2}}}$, $t\in [T-1]_{0}$, in order to obtain a comparable guarantee. 

\section{Case Study: Average Low-$K$ Utility Maximization}
\label{Sec:Case}
For better understanding of the specific application of SMWM, we consider in this section a specific class of welfare measure given in the following:
\begin{definition}[Low $K$-average Welfare]
Let be $K\in\nat$. We define the Lowest $K$-average utility as the mapping:
\begin{equation}
\label{Eq:ajjajsgsgsgffsfsgsfsfsss}
\underline{\Phi}_{K}:\real^{N}\rightarrow\real,\quad \bsu\mapsto  \tfrac{1}{K}\sum_{i=1}^{K}\bsu^{(\pi_{\bsu}(i))},
\end{equation} 	
where for any $u\in\real^{N}$, $\pi_{\bsu}:[N]\rightarrow [N]$ is a function satisfying $\bsu^{(\pi_{\bsu}(1))}\leq \bsu^{(\pi_{\bsu}(2))}\leq \cdots\leq \bsu^{(\pi_{\bsu}(N))}$
\end{definition}
In other words, the lowest $K$-average utility measures the average of the utilities of $K$ agents having the lowest utilities among all.

 Now, we show that this class of functionals is a subclass of welfare measures. The following proposition which is a straightforward application of the Karush-Kuhn-Tucker condition is helpful for this purpose: 
\begin{proposition}[Robust Representation of Low $K$-AW]
	\label{Prop:LowK}
Let $\underline{\Phi}_{K}$ be the Low-K average welfare.
\begin{enumerate}
\item It holds:
\begin{equation}
\label{Eq:ddkkrrjjrjrrrrrrr}
\underline{\Phi}_{K}(\bsu)=\min_{\bsw\in\A_{K}}\inn{\bsw}{\bsu},\quad \bsu\in\real^{D},
\end{equation}
where:
\begin{equation}
\label{Eq:aajjsssgsgggsssss}
\A_{K}:=\lrbrace{w\in\Delta([N]):~w_{i}\leq 1/K,~\forall i\in[D]},
\end{equation}
\item The solution of the optimization problem in \eqref{Eq:ddkkrrjjrjrrrrrrr} is given by:
\begin{equation}
\label{Eq:aaahhsgsgssfsfsfssgsgsgsss}
w_{i}^{*}=\begin{cases}
\tfrac{1}{K},\quad &\text{if } i\in\pi_{u}([K])\\
0,&\text{otherwise}
\end{cases},\quad \forall i\in [D].
\end{equation} 
\item $\partial\underline{\Phi}_{K}(0)=\A_{K}$, where $\A_{K}$ is given by \eqref{Eq:aajjsssgsgggsssss}
\end{enumerate}
\end{proposition}
\begin{algorithm}[htbp]
	\caption{Supergradient method for $K$-Low Average Welfare (SMWM $K$-Low)}
	\begin{algorithmic}[1]
		\Require
		time horizon $T\in\nat$, step-size sequence $(\gamma_{t})_{t\in [T-1]_{0}}$
		
		\For{$t=0,\ldots,T-1$}
		\For{all agents $i\in [N]$}
		\State Query $\bU^{(i)}(\boldnormtheta_{t})$
		\EndFor
		\State Find a permutation $\pi: [N]\rightarrow [N]$ satisfying:
		\begin{equation*}
		\bU^{(\pi(1))}(\boldnormtheta_{t})\leq \bU^{(\pi(2))}(\boldnormtheta_{t})\leq \ldots\leq \bU^{(\pi(N))}(\boldnormtheta_{t})  
		\end{equation*}
		\State Initialize $\bsg_{t}=0$
		\For{$i\in [K]$}
		\State Query $\tilde{\bsg}^{\pi(i)}_{t}\in\partial U_{\pi(i)}(\boldnormtheta_{t})$
		\State $\bsg_{t}\gets \bsg_{t}+\tilde{\bsg}^{\pi(i)}_{t}/K$
		\EndFor      
		\State Update:
		\begin{equation*}
		\boldnormtheta_{t+1}=\Pi_{\Theta}\left( \boldnormtheta_{t}+\gamma_{t}\bsg_{t}\right) 
		\end{equation*}

		\EndFor
		\State Take the ergodic average of $(\boldnormtheta_{t})_{t\in [T-1]_{0}}$ or the best iterate:
		\begin{equation*}
		\overline{\boldnormtheta}^{\gamma}_{T}=\tfrac{\sum_{t=0}^{T-1}\gamma_{t}\boldnormtheta_{t}}{\sum_{t=0}^{T-1}\gamma_{t}}\quad\text{or}\quad\boldnormtheta_{\max,T}\in\argmax_{t\in [T-1]_{0}}\Phi(\bU(\boldnormtheta_{t}))
		\end{equation*}
	\end{algorithmic}
	\label{Alg:aoaishhjddhhddddeee3}
\end{algorithm}
\setlength{\textfloatsep}{0pt}
Representation \eqref{Eq:ddkkrrjjrjrrrrrrr} and Theorem \ref{Thm:ajajssgsgsffsfsssssss} asserts that $\underline{\Phi}_{K}$ is indeed a welfare measure.

Respective to this class of measure, the resource allocation problem to solve is:
\begin{equation} 
\label{Eq:ajajassgsgssgshshsfffsgssss}
\max_{\boldnormtheta\in\Theta}\underline{\Phi}_{K}(\bU(\boldnormtheta)).
\end{equation}
The solution of above problem gives rise to the allocation strategy optimizing the total utility of the $K$ worst-off agents. Furthermore,  
notice that for $K=1$ above problem coincides with the max-min fairness problem \eqref{Eq:jaajjagsgsgsfsfssgsfssfsss}, and that for $K=N$ \eqref{Eq:ajajassgsgssgshshsfffsgssss} coincides with maximum average utilities problem \eqref{Eq:ajajjsgsgsfsffsfggsssssss}. Those observations assert that we can balance between \eqref{Eq:jaajjagsgsgsfsfssgsfssfsss} and \eqref{Eq:ajajassgsgssgshshsfffsgssss} by choosing $K$ between $1$ and $N$. For illustration, we provide a numerical simulations for this aspect in the next section.

Since $\underline{\Phi}_{K}$ is a welfare measure, we can use Algorithm \ref{Alg:aoaishhjddhhddddeee2} to solve \eqref{Eq:ajajassgsgssgshshsfffsgssss}. In the following, we specify the quite general steps of Algorithm \ref{Alg:aoaishhjddhhddddeee2} by specifying the computation of the supergradient of $\underline{\Phi}_{K}(\bU)$ for a given agents' utilities $\bU$. To do this, we use Theorem \ref{Thm:auauasgsgsfsfsssss} to compute a supergradient of $\underline{\Phi}_{K}$ at $\bU(\boldnormtheta)$. To this end, Theorem \ref{Thm:auauasgsgsfsfsssss} asserts to find $\bsw(\boldnormtheta)\in\partial \underline{\Phi}_{K}(0)$ for which $\Phi(\bU(\boldnormtheta))=\inn{\bsw_{\boldnormtheta}}{\bU(\boldnormtheta)}$. So, according to 2) in Proposition \ref{Prop:LowK}, a particular choice of such a weight vector is $\bsw^{(i)}(\boldnormtheta)=1/K$ if $i\in \pi_{\bU(\boldnormtheta)}([K])$ and $\bsw^{(i)}(\boldnormtheta)=0$ else. Now, we can present the corresponding algorithm, called the Supergradient method for $K$-Low Average Welfare (SMWM $K$-Low), in Algorithm \ref{Alg:aoaishhjddhhddddeee3}.

\section{Numerical Simulation: Low-K Maximization of Wireless Links QoS} 
\begin{figure}[htbp]
	\begin{center}
		\includegraphics[scale=0.7, trim={5.8cm 11.5cm 5.5cm 11.5cm}, clip]{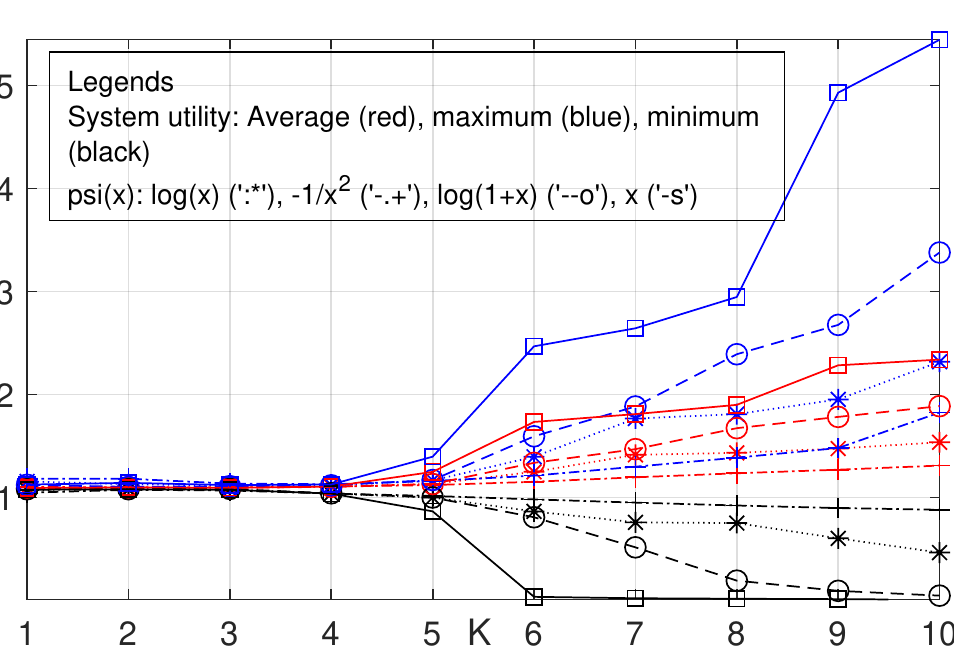}
	\end{center}
	\caption{Performance comparison of the $K$-low average of SMWM $K$-low (Algorithm \ref{Alg:aoaishhjddhhddddeee3}) respective the SINR for different $K\in [10]$ and $\psi$. 
	}
	\label{Fig:aoaojsjsjjddd1}
\end{figure}
\setlength{\textfloatsep}{0pt}  


To illustrate our results, we apply them to the specific application of power control in wireless network in Section \ref{Sec:ajjaasggsgsgsss}. Our interest is on solving the problem \eqref{Eq:ajajassgsgssgshshsfffsgssss}, where the parameter $\theta$ is the log transmission power of the agents, the constraint set $\Theta$ is equal to $\log\P$, with $\P$ denotes the power constraints \eqref{Eq:jjaaggsffsfsfsgsgsgsss}, the utility of the agent $i\in[N]$ is given by $\bU^{(i)}(\theta)=\psi(\SINR^{(i)}_{e}(\bss))$ with $\SINR$ is given as in \eqref{Eq:jaajjssgsgsgsgsfsfssss} and $\psi$ denotes the function specifying the QoS of the agents from their SINR (see \eqref{Eq:aajajssgsgsgsfsfsgssffsfsss}). 
In solving this problem, we use SMWM $K$-Low algorithm (Algorithm \ref{Alg:aoaishhjddhhddddeee3}). 
\paragraph*{System Model Parameters}
In our numerical simulations, we set the number of agents by $N=10$. We choose the communication gain parameters $V_{k,l}$, $k,l\in [10]$ randomly independently with the following specification for any $k\in [10]$: $V_{k,k}$ is uniformly distributed in the interval $[1,3]$ and $V_{k,l}$ is exponentially distributed with mean $1/10$. We set the noise power as $\sigma_{k}^{2}=1/5$ for all agents $k\in [10]$. 
\paragraph*{Optimization Parameters}  
The wireless network system in our simulations is subject to the power constraint \eqref{Eq:jjaaggsffsfsfsgsgsgsss} with $0.05\leq \bsp^{(k)}\leq 1$. Thus the constraint set $\log \P=\lrbrace{\bss\in\real^{N}:\bss^{(k)}\in [e^{0.05},1]}$. In our simulations we consider not only the choices of $\psi$ satisfying Assumption \ref{Ass:Log}, such as $\psi(x)=\log(x)$ and $\psi(x)=-1/x^{2}$, but also other choices of $\psi$ popular in applications such as $\psi(x)=\log(1+x)$ and $\psi(x)=x$.  

\paragraph*{Algorithmic Parameters}
The theoretical results in this paper asserts that the convexity of agents' utility is one ingredient for the success of SMWM $K$-Low. With $\psi$ satisfying Assumption \ref{Ass:Log} (e.g., $\psi(x)=\log(x)$ and $\psi(x)=-1/x^{2}$), one can show that the utility of agent $k$ given by $\bU^{(i)}(\bss)=\rmq^{(k)}_{e}(\bss)$ is a convex function in the resource allocation variable $\bss$. Nevertheless, we test for completeness also SMWM $K$-Low with $\psi$, which does not satisfy Assumption \ref{Ass:Log} but popular in practice. Furthermore, as the constraint set $\log \P$ specified in the last paragraph is convex and compact and by inspecting the Hessian of the utility functions, one can show that the requirements in Theorem \ref{Thm:akakasshshssgsggs} (provided that $\psi$ satisfies Assumption \ref{Ass:Log}) is fulfilled. Therefore, the success of SWMW is theoretically guaranteed. In our simulations, we consider the time horizon $T=2000$ and the fixed step size $\gamma_{t}=5/\sqrt{T}$, $t\in [T-1]_{0}$. Moreover, we always set initial iterate $\bss_{0}=0$. 

\begin{figure}[htbp]
	\begin{center}
		\includegraphics[scale=0.7, trim={5.0cm 11.5cm 5.5cm 11.5cm}, clip]{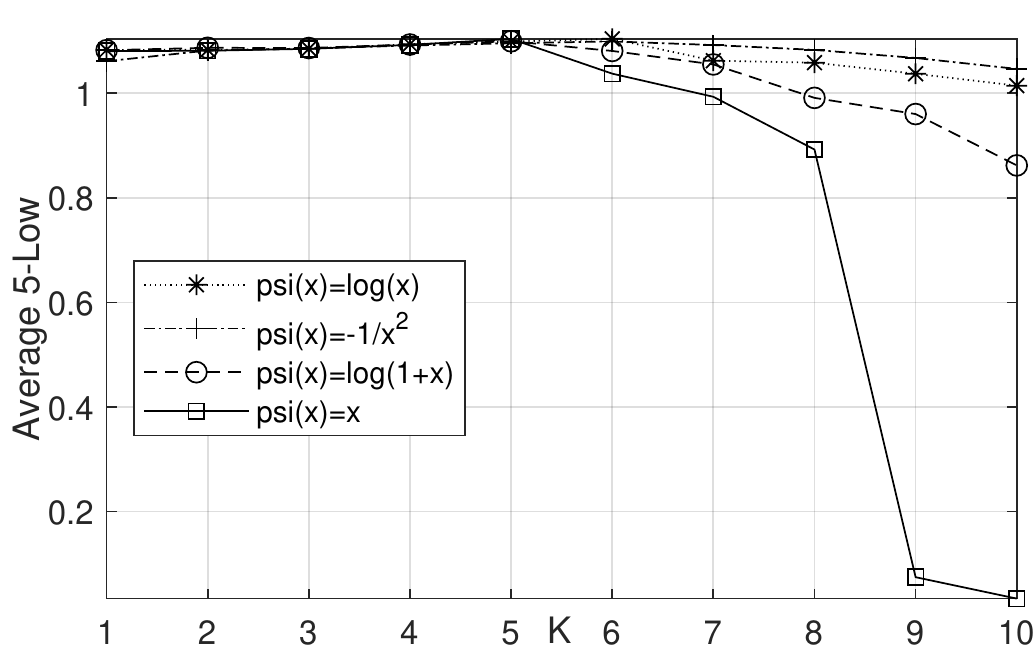}
	\end{center}
	\caption{Average $5$-low of the SINR comparison of SMWM $K$-low (Algorithm \ref{Alg:aoaishhjddhhddddeee3}) for different $K\in [10]$ and $\psi$. 
	}
	\label{Fig:aoaojsjsjjddd2}
\end{figure}
\setlength{\textfloatsep}{0pt}
\paragraph*{Simulation Results -- Average performance and Max-Min Fairness}
Figure \ref{Fig:aoaojsjsjjddd1} shows the performance of the output ($\boldnormtheta_{\max,T}$) of SMWM $K$-low, for different choices of averaging numbers $K\in [N]=[10]$ and $\psi$, in terms of the average (red lines), minimum (black lines), and maximum (blue lines) SINR of the agents'. The choices of $\psi$ range from those satisfying Assumption \ref{Ass:Log}, i.e. $\psi(x)=\log(x)$ (dotted lines with star markers) and $\psi(x)=-1/x^{2}$ (dotted-dashed lines with plus markers), and those not satisfying Assumption \ref{Ass:Log}, i.e. $\psi(x)=\log(1+x)$ (dashed lines with circle markers) and $\psi(x)=x$ (straight lines with square markers). One can see in Figure \ref{Fig:aoaojsjsjjddd1} that SWMW $K$-Low for $K=10$, corresponding to the max-average utility optimization, has the highest average SINR for any choices of $\psi$, which is to be expected as SWMW $10$-Low approximately provides the optimal resource allocation for the average utility. However, this superiority is of cost of inferiority of some agents' SINR, as the minimum of this quantity is at lowest for this choice of $K$. With decreasing $K$, we observe the tendency of the tradeoff in form of decreasing total SINR and increasing minimum of agents' SINR. However, the performance of SWMW $1$-low in our simulation is (slightly) sub-optimal as it not yields the maximum lowest utility of the agents upon all choices of $K$. This might be due to lack of (strong) convexity making the corresponding supergradient method slow. To solve the optimal resource allocation problem respective to the average $1$-low utility, one may either increase the time horizon or use another method given in the literature. Based on the numerical observation, one may alternatively use SWMW with small $K\neq 1$.

\paragraph*{Simulation Results-- Optimality of Average $K$-Low}
To check whether SWMW $K$-low produces optimal resource allocation for average $K$-low of the utilities, we check the performance of the output of SWMW $K$-low respective to the average $5$-low as system performance measure. We plot our result in Figure \ref{Fig:aoaojsjsjjddd2}. There, we observe that irrespective of the choice of $\psi$, the highest value of the average $5$-low utilities is achieved by utilizing SWM $K$-low for $K=5$ verifying our theoretical result SWMW $K$-low produces an approximate solution of the problem \eqref{Eq:ajajassgsgssgshshsfffsgssss}.

\section{Conclusion}
In this paper, we have presented a general notion of a performance/welfare measure of a resource-driven competitive multi-agent system. This gives rise not only to the popular system optimality notions, such as the popular average optimality and the max-min fairness, but also to interesting non-standard system optimality notions, such as the optimality of the subset of worst-off agents. One clear advantage of the latter which is particularly obvious from our numerical investigations is that it provides an alternative tradeoff between the aforementioned popular optimality notions. Furthermore, we were able, by means of convex analytical method, to relate an abstract welfare measure to the popular average/total performance measure. This provides a way to extend techniques given in the literature using the latter performance measure, so that they can handle resource allocation objective respective to the former. One particular example given in this work is the supergradient method for seeking an optimal allocation respective to a general welfare principle (SMWM). Interesting directions for the future are investigations on the structure of Low $K$-average welfare optimization in dependence of the communication gain matrix, and on the distributed implementation and acceleration of the first-order algorithms given in this work.  

\appendix
\subsection{Basic Notions and Notations for proofs }
In the proof, we make use of the indicator function of a convex set:
\begin{equation*}
\delta_{\X}(\bsx)=\begin{cases}
0,\quad &\bsx\in\X\\
-\infty, &\bsx\notin\X
\end{cases}
\end{equation*}
Furthermore in the proofs we heavily make use of the following concept:
\begin{definition}[Concave Conjugate and biconjugate]
	Let be $\psi:\real^{D}\rightarrow\overline{\real}$ be proper. The concave conjugate of $\psi$ is defined as $\psi^{*}:\real^{D}\rightarrow\overline{\real}$ given by:
	\begin{equation*}
	\psi^{*}(y)=\inf_{\bsx\in\real^D{}}\lrbrace{\inn{\bsx}{\bsy}-\psi(\bsx)}.
	\end{equation*}
	The biconjugate $\psi^{**}$ of $\psi$ is defined as the concave conjugate of $\psi^{*}$. 
\end{definition}
Helpful for our approach is the following well-known facts in convex analysis (see e.g.,\cite{Rockafellar1970})
\begin{proposition}
	\label{Prop:Folklore}
	Let be $\psi:\real^{D}\rightarrow\overline{\real}$ be proper, concave, and upper semi-continuous. Then $\psi^{*}$ is proper and concave.
	\begin{enumerate}
		\item $\psi^{**}=\psi$
		\item The following statements are equivalent:
		\begin{enumerate}
			\item $\inn{\bsx}{\bsy}=\rmf(\bsx)+\rmf^{*}(\bsy)$
			\item $\bsy\in\partial \rmf(\bsx)$
			\item $\bsx\in\partial \rmf^{*}(\bsy)$
		\end{enumerate}
	\end{enumerate}
\end{proposition}
\subsection{Missing Proofs in Section \ref{Sec:ajajgsgsgsffssgsfsgsfgssss}}
\label{Subsec:Proof}
\begin{proof}[Proof of Theorem \ref{Thm:ajajssgsgsffsfsssssss}]
	If \eqref{Eq:ajajajshsshsggsgsgsgsssss} holds, then it is straightforward to show that $\Phi$ is a welfare measure. Now, we show the reverse statement.

	Suppose that $\Phi$ is a welfare measure. As $\Phi$ is a proper upper semi-continuous convex function, it follows that:
	\begin{equation*}
	\Phi(\bsu)=\Phi^{**}(\bsu)=\inf_{\bsw\in\real^{N}}\lrbrace{\inn{\bsw}{\bsu}-\Phi^{*}(\bsw)},
	\end{equation*}
	where the first equality follows from Fenchel-Moreau theorem (see $1)$ in Proposition \ref{Prop:Folklore}) and the second equality from the definition of the concave conjugate of a function (here: $\Phi^{*}$).
	
	As $\Phi$ is homogeneous, it follows that:
	\begin{equation}
	\label{Eq:aajsssggssfsfsfssss}
	\Phi^{*}(\bsw)=\delta_{\partial\Phi(0)}(\bsw),
	\end{equation}
	Indeed, for any $\lambda> 0$, it holds:
	\begin{equation*}
	\begin{split}
	&\Phi^{*}(\bsw)=\inf_{\bsu\in\real^{N}}\lrbrace{\inn{\bsu}{\bsw}-\Phi(\bsu)}=\inf_{\bsu\in\real^{N}}\lrbrace{\inn{\bsu}{\lambda \bsw}-\Phi(\lambda \bsu)}\\
	&=\inf_{\bsu\in\real^{N}}\lrbrace{\lambda\inn{\bsu}{ \bsw}-\lambda\Phi( \bsu)}=\lambda\Phi^{*}(\bsw),
	\end{split}
	\end{equation*}
	where the second equality follows from the change of optimization variable $\bsu\rightarrow\lambda \bsu$ preserving the optimization problem as $\lambda>0$, and the third equality from the assumption that $\Phi$ is homogeneous. Letting $\lambda\rightarrow\infty$ in above equation, it follows that if $\Phi^{*}(\bsw)\neq 0$ then either $\Phi^{*}(\bsw)=\infty$ or $\Phi^{*}(\bsw)=-\infty$. The former case can not occur as $\Phi^{*}$ is proper (see Proposition \ref{Prop:Folklore}). Thus for any $\bsu\in \real^{N}$, either $\Phi^{*}(\bsw)=0$ or $\Phi^{*}(\bsw)=\infty$. So, for showing \eqref{Eq:aajsssggssfsfsfssss}, it remains to specify the set on which $\Phi^{*}$ takes zero values. For this sake, notice that as $\Phi$ is homogeneous, we have that $\Phi(0)=0$. This asserts, that $\bsw\in\real^{N}$ satisfies $\Phi^{*}(\bsw)=0$ if and only if $\inn{0}{\bsw}=\Phi(0)+\Phi^{*}(\bsw)$. Consequently, by $2)$ in Proposition \ref{Prop:Folklore}, this holds if and only if $\bsw\in\partial \Phi(0)$. 
	
	As the consequence of \eqref{Eq:aajsssggssfsfsfssss}, we have that $\Phi$ has the representation \eqref{Eq:ajajajshsshsggsgsgsgsssss}, where $\U$ is given explicitly by $\partial\Phi(0)$. It is well-known that the superdifferential set of a function at a point in the domain of the function is a non-empty closed and convex set. Therefore as $\Phi$ is a real-valued function and thus has the domain equal to the whole $\real^{N}$, to show the remaining statement in $1)$, we need to establish the fact that $\partial\Phi(0)$ is a subset of the simplex. For this sake, take an arbitrary $\bsw\in\partial\Phi(0)$. We have by the definition of the supergradient, and the fact that $\Phi(0)=0$ following from the positive homogeneity of $\Phi$:
	\begin{equation*}
	\label{Eq:ajajashssggsgsgsssss}
	\Phi(\pm 1)\leq \Phi(0)+\inn{\bsw}{\pm 1-0}=\inn{\bsw}{\pm 1}
	\end{equation*}  
	Thus, since the welfare measure $\Phi$ satisfies $\Phi(\pm 1)=\pm\Phi(1)$, we have $\sum_{i=1}\bsw^{(i)}=1$. Furthermore for any $i\in[N]$, we have by replacing $\pm 1$ in \eqref{Eq:ajajashssggsgsgsssss} by $e_{i}$ and by noticing that monotonicity of $\Phi$ yields $\Phi(e_{i})\geq 0$, that $\bsw^{(i)}\geq 0$, yielding the fact that $\bsw\in\Delta([N])$.
	
	At last we show the uniqueness statement in $2)$. For this sake, suppose that the welfare measure $\Phi$ can be represented as in \eqref{Eq:ajajajshsshsggsgsgsgsssss} by a non-empty closed convex subset $\U\subseteq\Delta([N])$ other than $\partial\Phi(0)$. Then, we have for any $\bsu\in\real^{N}$:
	\begin{equation*}
	(\delta_{\U})^{*}(\bsu)=\min_{\bsw\in\U}\inn{\bsu}{x}=\min_{x\in\tilde{\partial\Phi(0)}}\inn{\bsu}{x}=(\delta_{\partial\Phi(0)})^{*}(\bsu),
	\end{equation*}
	where the second equality follows from the previously proven fact that the welfare measure $\Phi$ can be represented as in \eqref{Eq:ajajajshsshsggsgsgsgsssss}, where the optimization is over the set $\partial\Phi(0)$.  
	Consequently $\delta_{\U}^{**}=\delta_{\partial\Phi(0)}^{**}$. As $\U$ and $\partial\Phi (0)$ are convex and closed, it follows that $\delta_{\U}$ and $\delta_{\partial\Phi(0)}$ are concave and upper semi-continuous. Consequently, we have by Fenchel-Moreau Theorem (see 1) in Proposition 	\ref{Prop:Folklore}), that $\delta_{\U}^{**}=\delta_{\U}$ and $\delta_{\partial\Phi (0)}^{**}=\delta_{\partial \Phi(0)}$. Combining all the results, we have $\delta_{\U}=\delta_{\partial\Phi(0)}$ contradicting with the assumption $\U\neq\partial\Phi(0)$ 
\end{proof}

\begin{proof}[Proof of Lemma \ref{Lem:aajssgsgsgsfsfsss}]
	Let $\bsu\in\real^{D}$ be arbitrary. By 2) in Proposition \ref{Prop:Folklore}, we have that $w\in\partial\Phi(\bsu)$ if and only if:
	\begin{equation}
	\label{Eq:ahahssgsfsfsfsgsffsssfss}
	\inn{\bsu}{\bsw}=\Phi(\bsu)+\Phi^{*}(\bsw)=\Phi(\bsu)+\delta_{\partial\Phi(0)}(\bsw),
	\end{equation}
	where we use for the second equality the identity $\Phi^{*}=\delta_{\partial\Phi(0)}$ shown in Theorem \ref{Thm:ajajssgsgsffsfsssssss}. From \eqref{Eq:ahahssgsfsfsfsgsffsssfss}, it follows that in order $\bsw\in\partial\Phi(\bsu)$, it is necessary that $\bsw\in\partial \Phi (0)$. For this kind of $\bsw$, \eqref{Eq:ahahssgsfsfsfsgsffsssfss} yields that $\bsw\in\partial\Phi(\bsu)$ if and only if $\inn{\bsw}{\bsu}=\Phi(\bsu)$. 
\end{proof}
\section{Missing Proofs in Section \ref{Sec:SupGrad}}
\label{Sec:}
\begin{proof}[Proof of Theorem \ref{Thm:auauasgsgsfsfsssss}]
	The fact that $\Phi(\bU)$ is concave is an implication of the fact that $\Phi(\bU)$ is the pointwise minimum of concave functions (see \eqref{Eq:ajajajshsshsggsgsgsgsssss}). The fact that $\Phi(\bU)$ is proper is clear.

	It is well known that the subdifferential of a proper convex function on the relative interior (which is also non-empty) of the domain of the function is non-empty. Thus above statement is not a vacuous truth. Now, it follows Lemma \ref{Lem:aajssgsgsgsfsfsss} that $w\in\partial \Phi(0)$ satisfying \eqref{Eq:aahahsgsgsfsfssgfsssgsgs} is contained in $\partial_{\bU(\boldnormtheta)} \Phi(\bU(\boldnormtheta))$. Consequently by the definition of supergradient: 
	\begin{equation}
	\label{Eq:ajajasshshshsggsgsshshgshsss}
	\Phi(\bU(\tilde{\boldnormtheta}))\leq \Phi(\bU(\boldnormtheta))+\inn{g}{\bU(\tilde{\boldnormtheta})-\boldnormtheta(\boldnormtheta)}.
	\end{equation}
	Moreover by Theorem \ref{Thm:ajajssgsgsffsfsssssss}, we have that $\partial\Phi(0)\subseteq\Delta([N])$. As a consequence, we have that $w\geq 0$. This and the fact that $\tilde{g}^{(i)}\in\partial \bU^{(i)}(\boldnormtheta)$ yields:
	\begin{equation}
	\label{Eq:ajajasshshshsggsgsshshgshsss2}
	\begin{split}
	\inn{\bsw}{\bU(\tilde{\boldnormtheta})-\bU(\boldnormtheta)}&=\sum_{i=1}^{N}\bsw^{(i)}(\bU^{(i)}(\tilde{\boldnormtheta})-\bU^{(i)}(\boldnormtheta))\\
	&\leq \inn{\sum_{i=1}^{N}\bsw^{(i)}\tilde{\bsg}^{(i)}}{\tilde{\boldnormtheta}-\boldnormtheta}
	\end{split}
	\end{equation}
	Combining \eqref{Eq:ajajasshshshsggsgsshshgshsss} and \eqref{Eq:ajajasshshshsggsgsshshgshsss2}, we obtain the $\Phi(\bU(\tilde{\boldnormtheta}))\leq \Phi(\bU(\boldnormtheta))+ \inn{\sum_{i=1}^{N}\bsw_{i}\tilde{\bsg}^{(i)}}{\tilde{\boldnormtheta}-\boldnormtheta}$, showing that $\sum_{i=1}^{N}\bsw^{(i)}\tilde{\bsg}^{(i)}\in\partial_{\boldnormtheta}\Phi(\bU(\boldnormtheta))$. 
\end{proof}
\begin{proof}[Proof of Lemma \ref{Lem:bounded}]
	We have:
	\begin{equation*}
	\norm{\bsg_{t}}_{2}=\norm{\sum_{i=1}^{N}\bsw^{(i)}_{t}\tilde{\bsg}_{t}^{(i)}}_{2}\leq \sum_{i=1}^{N}\bsw^{(i)}_{t}\norm{\tilde{\bsg}_{t}^{(i)}}_{2}. 
	\end{equation*}
	As $\tilde{\bsg}_{t}^{(i)}$ is a supergradient of $\bU_{i}$, we have by the uniform boundedness assumption:
	\begin{equation*}
	\sum_{i=1}^{N}\bsw^{(i)}_{t}\norm{\tilde{\bsg}_{t}^{(i)}}_{2}\leq \sum_{i=1}^{N}\bsw^{(i)}_{t}\bsM^{(i)}.
	\end{equation*} 
	As $\bsw_{t}\in\partial\Phi(0)$, we obtain the desired statement by combining above inequalities and by taking the corresponding maximum.
\end{proof}

\begin{proof}[Proof Proposition \ref{Prop:LowK}]
	The KKT condition for the optimization problem is given by:
	\begin{equation}
	\label{Eq:ahjahahssgsfsfsgfsfsfsssss}
	\begin{split}
	&\bsu^{(i)}+\lambda_{*}^{(i)}-\mu_{*}^{(i)}+\eta_{*}=0,~\forall i\in [D]\\
	&\lambda_{*}^{(i)}\left( \bsw_{*}^{(i)}-\tfrac{1}{K}\right) =0,~\mu_{*}^{(i)}\bsw_{*}^{(i)}=0,~\forall i\in [D]\\
	&\lambda_{*}^{(i)}\geq 0,~\mu_{*}^{(i)}\geq 0,~\forall i\in [D].
	\end{split}
	\end{equation}
	As the optimization problem \eqref{Eq:ddkkrrjjrjrrrrrrr} is a linear problem, and therefore convex, it follows that for $(w_{*},\lambda_{*},\mu_{*},\eta_{*})\in\real^{D}\times\real^{D}\times\real^{D}\times \real$, $w_{*}$ is a solution of \eqref{Eq:ddkkrrjjrjrrrrrrr}. Let be $w_{*}$ given by \eqref{Eq:aaahhsgsgssfsfsfssgsgsgsss}, $\eta_{*}:=u_{\pi_{u}(K)}$, $\lambda_{*}$ and $\mu_{*}$ given by 
	\begin{equation*}
	\begin{split}
	&\lambda^{(i)}_{*}:=\bsu^{(\pi_{u}(K))}-\bsu^{(i)}\quad\text{and}\quad \mu^{(i)}_{*}:=0,\quad \forall i\in \pi_{\bsu}([K])\\
	&\lambda^{(i)}_{*}:=0\quad\text{and} \quad\mu^{(i)}_{*}:=\bsu^{(i)}\bs-u^{(\pi_{\bsu}(K))},\quad \forall i\notin \pi_{\bsu}([K])
	\end{split}.
	\end{equation*}
	Immediately, one checks that this choice of $(\bsw_{*},\lambda_{*},\mu_{*},\eta_{*})$ satisfies \eqref{Eq:ahjahahssgsfsfsgfsfsfsssss} and thus $\bsw_{*}$ is a solution of the optimization problem given by \eqref{Eq:ddkkrrjjrjrrrrrrr}. Setting this $\bsw_{*}$ into the objective of the problem \eqref{Eq:ddkkrrjjrjrrrrrrr}, we obtain the identity in \eqref{Eq:ddkkrrjjrjrrrrrrr}. Now, it remains to show the last statement. Clearly, $\A$ is non-empty, closed, and convex. Consequently, the representation \eqref{Eq:ahjahahssgsfsfsgfsfsfsssss} and 2) in Theorem \ref{Thm:ajajssgsgsffsfsssssss} asserts that $\partial\underline{\Phi}_{K}(0)=\A$, as desired.
\end{proof}

\end{document}